\documentclass[aps,twocolumn]{revtex4-2}

\usepackage[utf8]{inputenc}
\usepackage[margin=2.5cm]{geometry}
\usepackage{amsmath}
\usepackage{amsfonts}
\usepackage{enumerate}
\usepackage{amsthm}
\usepackage{braket}
\usepackage{algorithm}
\usepackage{algpseudocode}
\usepackage{appendix}
\usepackage{hyperref}
\usepackage[capitalise]{cleveref}
\usepackage{multirow}
\usepackage{booktabs}
\usepackage{tabularray}
\UseTblrLibrary{booktabs}
\usepackage{xcolor}
\usepackage{graphicx}
\usepackage{autonum}
\usepackage{mathrsfs}  
\usepackage{subcaption}
\usepackage{ragged2e}

\theoremstyle{plain}
\newtheorem{theorem}{Theorem}[]
\newtheorem{lemma}[theorem]{Lemma}
\newtheorem{corollary}[theorem]{Corollary}
\newtheorem{claim}[theorem]{Claim}

\newcommand{\bb}[1]{\mathbb{#1}}

\newcommand{\cH}{\mathcal{H}}
\newcommand{\cB}{\mathcal{B}}
\newcommand*{\proj}[1]{|#1 \rangle\langle #1|}

\DeclareMathOperator{\diag}{diag}
\newcommand{\cat}{\mathrm{cat}}
\newcommand{\GKP}{\mathrm{GKP}}
\renewcommand{\braket}[1]{\langle#1\rangle}

\newcommand*{\tr}{\mathsf{tr}}

\makeatletter
\newcommand{\subalign}[1]{%
  \vcenter{%
    \Let@ \restore@math@cr \default@tag
    \baselineskip\fontdimen10 \scriptfont\tw@
    \advance\baselineskip\fontdimen12 \scriptfont\tw@
    \lineskip\thr@@\fontdimen8 \scriptfont\thr@@
    \lineskiplimit\lineskip
    \ialign{\hfil$\m@th\scriptstyle##$&$\m@th\scriptstyle{}##$\hfil\crcr
      #1\crcr
    }%
  }%
}
\makeatother

\begin{document}

\author{Beatriz Dias}
\author{Robert K{\"o}nig}
\affiliation{Department of Mathematics, School of Computation, Information and Technology, Technical University of Munich, Garching, Germany}
\affiliation{Munich Center for Quantum Science and Technology, Munich, Germany}
\date{\today}

\title{On the sampling complexity of coherent superpositions}

\begin{abstract}
We consider the problem of sampling from the distribution of measurement outcomes when applying a POVM to a superposition~$\ket{\Psi} = \sum_{j=0}^{\chi-1} c_j \ket{\psi_j}$ of~$\chi$ pure states. We relate this problem to that of drawing samples from the outcome distribution when measuring a single state~$\ket{\psi_j}$ in the superposition. Here~$j$ is drawn from the distribution~$p(j)=|c_j|^2/\|c\|^2_2$ of normalized amplitudes. We give an algorithm which -- given~$O(\chi \|c\|_2^2 \log1/\delta)$ such samples and calls to oracles evaluating the involved probability density functions -- outputs a sample from the target distribution except with probability at most~$\delta$. 

In many cases of interest, the POVM and individual states in the superposition have efficient classical descriptions  allowing to evaluate matrix elements of POVM elements and to draw samples from outcome distributions. In such a scenario, our algorithm gives a reduction from strong classical simulation (i.e., the problem of computing outcome probabilities) to weak simulation (i.e., the problem of sampling). In contrast to prior work focusing on finite-outcome POVMs, this reduction also applies to  continuous-outcome POVMs. An example is homodyne or heterodyne measurements applied to a superposition of Gaussian states. Here we obtain a sampling algorithm with time complexity~$O(N^3 \chi^3 \|c\|_2^2 \log1/\delta)$ for a state of~$N$ bosonic modes.
\end{abstract}

\maketitle

The superposition principle is a cornerstone of quantum mechanics with profound implications expressed in a plethora of experimental signatures. These include interference patterns observed in the double slit experiment \cite{PhysRev.30.705},  Rabi oscillations appearing in two-level systems in an oscillating field \cite{PhysRev.51.652}, as well as key features of quantum correlations and non-locality such as Bell inequality violations \cite{PhysicsPhysiqueFizika.1.195}. Coherent superpositions of states also figure prominently in all  proposals towards demonstrating an algorithmic quantum advantage for computing:  While the role of ``quantum parallelism'' as a source of quantum computational power is at times inflated in the popular science literature, both  coherent application of function evaluation (e.g., in oracle calls in Grover's algorithm \cite{10.1145/237814.237866} or modular power computation in Shor's algorithm \cite{10.1109/SFCS.1994.365700}), as well as constructive and destructive interference are arguably the most essential building blocks in quantum algorithms design.

Given that coherent superpositions give rise to quantum computational power, it is natural to try to quantify  the algorithmic complexity
associated with coherence. With this objective, we consider a particular algorithmic problem, namely that of sampling from the output distribution when measuring a state. Our main result is an algorithmic recipe which provides an upper bound on the resources required to draw such a sample. 
In more detail, let~$\cH$ be a Hilbert space, $(\Omega,\Sigma)$  a measure space and~$M:\Sigma\rightarrow\cB(\cH)$ a positive operator-valued measure (POVM) on~$\cH$.  According to Born's rule, applying the measurement described by the POVM~$M$ to a state~$\rho\in \cH$ yields a measurement result~$\zeta^M_{\rho}$ which is a sample from the probability distribution~$\mu^M_{\rho}$ defined by
\begin{align}
\Pr\left[\zeta^M_{\rho}\in X\right]&=\mu^M_{\rho}(X):=\tr(M(X)\rho)\label{eq:bornsrule}
\end{align}
for every measurable set~$X\in\Sigma$. 
Our goal is to assess the complexity of drawing a sample~$\zeta^M_\Psi\sim \mu^M_\Psi$ for the case where
\begin{align}
\ket{\Psi}&=\sum_{j=0}^{\chi-1} c_j \ket{\psi_j}\ \label{eq:coherentsuperposition}
\end{align} 
is a superposition of~$\chi$ states~$\{\ket{\psi_j}\}_{j=0}^{\chi-1}$ with complex coefficients~$\{c_j\}_{j=0}^{\chi-1}$. We note that 
the latter two objects are related by  the normalization condition~$1=\|\Psi\|^2=\sum_{j,k=0}^{\chi-1}\overline{c_j}c_k\langle \psi_j,\psi_k\rangle$.

 We are particularly interested in the increase of complexity which needs to be attributed to the fact that we are dealing with a coherent superposition (see Eq.~\eqref{eq:coherentsuperposition}) rather than a mixture of the individual states~$\ket{\psi_j}_{j=0}^{\chi-1}$. For comparison, consider the simpler problem of drawing a sample~$\zeta^M_\rho\sim \mu^M_\rho$  from the outcome distribution~$\mu^M_\rho$ obtained when measuring an incoherent mixture, that is, the ensemble average~$\rho=\sum_{j=0}^{\chi-1}p(j)\proj{\psi_j}$  of the states~$\{\ket{\psi_j}\}_{j=0}^{\chi-1}$, where~$p$ is a probability distribution on the set~$\{0,\ldots,\chi-1\}$.  By linearity of Born's rule~\eqref{eq:bornsrule}, the outcome distribution in this case is simply the convex combination~$\mu^M_\rho=\sum_{j=0}^{\chi-1} p(j)\mu^M_{\psi_j}$ of the individual distributions~$\{\mu^M_{\psi_j}\}_{j=0}^{\chi-1}$ associated  with measuring a single state~$\ket{\psi_j}$.
 This suggests a simple recipe for drawing a sample~$\zeta^M_\rho$: Simply draw an element~$J\sim p$ from the set~$\{0,\ldots,\chi-1\}$ according to the distribution~$p$, and subsequently draw a sample~$\zeta^M_{\Psi_J}$
according to the distribution defined by measuring the state~$\ket{\Psi_J}$. Stated succinctly, this means that a single sample
\begin{align}
(J,X)\quad \textrm{where}\quad J\sim q\quad \textrm{and}\quad X\sim \mu^M_{\Psi_J}
\end{align} is sufficient to produce one sample from~$\mu^M_\rho$.

Drawing a sample from the outcome distribution~$\mu^M_\Psi$  when measuring a superposition~$\ket{\Psi}$ as in~\eqref{eq:coherentsuperposition} is more challenging since this distribution is not simply a convex combination of the individual distributions~$\{\mu^M_{\psi_j}\}_{j=0}^{\chi-1}$. Instead, we have
\begin{align}
\mu^M_{\Psi}(X)&=\sum_{j,k=0}^{\chi-1}\overline{c}_j c_k \langle \psi_j,M(X)\psi_k\rangle
\end{align}
for~$X\in \Sigma$, which means that different terms in the superposition can interfere both constructively and destructively. Nevertheless, we argue (see Theorem~\ref{thm:main} below) that sufficiently many i.i.d.~samples~$(J_1,X_1),\ldots,(J_n,X_n)$, where 
\begin{align}
(J_r,X_r)\quad \textrm{where}\quad J\sim q\quad \textrm{and}\quad X\sim \mu^M_{\Psi_J}\ ,\label{eq:marginaljrxr}
\end{align}
and where~$p$ is the distribution of normalized amplitudes defined by~$p(j)=|c_j|^2/\|c\|_2^2$, are sufficient to produce a sample from~$\mu^M_\Psi$. This is assuming that corresponding probability density functions can be evaluated (e.g., by querying oracles or efficient classical computation). We argue that a number~$n=O(\chi \|c\|_2^2)$ of samples scaling linearly in the number~$\chi$ of terms in the superposition as well as the squared~$2$-norm~$\|c\|_2^2$ of the coefficients is sufficient  to produce a sample~$\zeta^M_\Psi\sim \mu^M_\Psi$ with constant probability.

While our techniques and results can immediately be adapted to finite and countably infinite sets of outcomes, our focus is on the case where the set~$\Omega$ of outcomes is continuous. To treat this case, we will assume here that the POVM~$M$ satisfies the Radon-Nikodym property, i.e., there is  
a~$\sigma$-finite measure~$\lambda$ on~$(\Omega,\Sigma)$ and a~$\lambda$-measurable function~$Q:X\rightarrow\cB(\cH)$ with~$Q(x)\geq 0$ for all~$x\in \Omega$ such that 
\begin{align}
 M(X)=\int_X Q(x)d\lambda(x)\quad\textrm{for all}\quad X\in\Sigma\ .
\end{align}
We refer to Appendix~\ref{sec:radonnikodymPOVM} for more details, including a discussion of necessary and sufficient conditions for this Radon-Nikodym property to hold. It implies that  the probability distribution~$\mu^M_{\Psi}$ of outcomes when measuring a state~$\ket{\Psi}\in\cH$  (see Eq.~\eqref{eq:bornsrule}) is absolutely continuous with respect to~$\lambda$ with Radon-Nikodym derivative
\begin{align}
\frac{d\mu^M_{\Psi}}{d\lambda}(x)&=\langle \Psi,Q(x)\Psi\rangle=:f^M_\Psi(x)\quad\textrm{for}\quad x\in \Omega\ .\label{eq:rnderiv}
\end{align}

Now consider the problem of sampling from~$\mu^M_\Psi$, where~$\ket{\Psi}$ is a superposition as in~\eqref{eq:coherentsuperposition}. We assume that information about the instance~$(M,\ket{\Psi})$ is provided 
in the form of access to  oracles evaluating 
the function~$f^M_\Psi$, each  function~$f^M_{\psi_j}$ with~$j\in \{0, \ldots, \chi-1\}$, the function~$p(j)=|c_j|^2/\|c\|_2^2$, as well as i.i.d.~samples with distribution as in~\eqref{eq:marginaljrxr}.

\begin{theorem}\label{thm:main}
Let~$n\in\mathbb{N}$.  
We give an algorithm which produces a sample~$\zeta^M_{\Psi}\sim \mu^M_\Psi$  from the distribution of measurement outcomes when applying the POVM~$M$ to the state~$\ket{\Psi}$, except with probability at most ~$\exp(-n/(\chi \|c\|_2^2))$ (in which case it outputs ``FAIL'').
The algorithm runs in time~$O(n)$, and uses at most~$n$
\begin{align}
    (A)&\textrm{ oracle calls to }
    \begin{cases}
        \textrm{(A1) the function~$f^M_{\Psi}(\cdot)$}\\
        \textrm{(A2) each function~$\{f^M_{\psi_j}(\cdot)\}_{j=0}^{\chi-1}$}\\ 
        \textrm{(A3) the function $p(\cdot)$, and}
    \end{cases}\\
        (B) & \textrm{ i.i.d.~samples~$\{(J_r,X_r)\}_{r=1}^n$ as in Eq.~\eqref{eq:marginaljrxr}. }\hfill
\end{align}
\end{theorem}

In important special cases we describe below, the oracles~
(A1)--(A3) 
can be realized by efficient classical (typically deterministic) algorithms. Similarly,~(B) 
can be realized, e.g., by a pair~$(\mathsf{samp}_p,\mathsf{samp}^M)$ of efficient classical randomized algorithms, where 
\renewcommand{\theenumi}{$\widetilde{B}$\arabic{enumi}} 
\begin{enumerate} 
\item  $\mathsf{samp}_p$ outputs a random sample~$J\sim p$ distributed according to~$p$,  whereas \label{it:samplingqalgorithm}
\item $\mathsf{samp}^M$, on input~$j\in \{0,\ldots,\chi-1\}$ (or equivalently a classical description of~$\psi_j$), produces a sample~$X\sim \mu^M_{\psi_j}$.
\end{enumerate}
Clearly, calling this pair of oracles~$n$ times produces samples as assumed in~(B).

The construction of a suitable algorithm~$\mathsf{samp}^M$ needs to exploit
specific structural features of  the considered family~$\{\psi_j\}_{j=0}^{\chi-1}$ of states and the POVM~$M$. In contrast, a generic algorithm~$\mathsf{samp}_p$ satisfying~\eqref{it:samplingqalgorithm} can be constructed using~$O(\chi \log \chi)$ calls to the oracle~$p(\cdot)$ (and identical runtime): Using a binary representation~$J\equiv (J_{\lceil \log \chi\rceil},\cdots, J_0)$, $J\sim p$ can be sampled from in a bit-wise fashion by using a biased coin flip to draw the~$r$-th bit~$J_r$ according to the conditional probability~$\Pr\left[J_r=1\ |\ J_0=j_0,\ldots,J_{r-1}=j_{r-1}\right]$. The latter probability can be computed using Bayes' rule and~$O(\chi)$ calls to the oracle~$p(\cdot)$. 

As an example, we obtain the following corollary. We denote by 
$\mathrm{vac}(x) = e^{-x^2/2} / \pi^{1/4}$ the single-mode vacuum state and by 
$\ket{\psi_{\Gamma,d}}$ the Gaussian pure state on~$N$ bosonic modes with covariance matrix~$\Gamma$ and displacement vector~$d$, with global phase such that~$\langle \psi_{\Gamma,d} , \mathrm{vac} \rangle^{\otimes N} > 0$. Finally, let `$\mathrm{het}$'~denote heterodyne detection on each mode. 
\begin{corollary}[Heterodyne measurements]\label{cor:gaussiansuperpositionsampling}
 There is a classical probabilistic algorithm which takes as input a tuple~$\{(c_j,\Gamma_j,d_j)\}_{j=0}^{\chi-1}$
describing a superposition~$\ket{\Psi}=\sum_{j=0}^{\chi-1} c_j \ket{\psi_{\Gamma_j,d_j}}$ of Gaussian states on~$N$ bosonic modes, and outputs a sample~$\zeta^\mathrm{het}_\Psi\sim \mu^\mathrm{het}_\Psi$ with probability at least~$1-\delta$. The algorithm has  runtime~$O(N^3\chi^3\|c\|_2^2 \log1/\delta)$. 
\end{corollary}
\begin{proof}
Given the input~$\{(c_j,\Gamma_j,d_j)\}_{j=0}^{\chi-1}$ the algorithm in \cite[Lemma 5.1]{PhysRevA.110.042402} evaluates the function~$f_\Psi^\mathrm{het}$ in time~$O(\chi^2 N^3)$. 
The function~$p(j)=|c_j|^2/\|c\|_2^2$ can be evaluated in time~$O(1)$ for any input~$j\in \{0,\ldots,\chi-1\}$ after precomputing~$\|c\|^2_2$, which can be done in time~$O(\chi)$. Sampling from~$p$ can be achieved with~$O(\chi \log \chi)$  calls to~$p(\cdot)$ using the generic algorithm~$\mathsf{samp}_p$ described in the main text.  For each
$j\in\{0,\ldots, \chi-1\}$, the distribution~$\mu_{\psi_{\Gamma_j,d_j}}^\mathrm{het}$
is  Gaussian with mean vector~$d_j$ and covariance matrix~$(\Gamma_j + I)/2$ (see Appendix~\ref{sec:gauss_measurement}). Thus evaluating
$f_{\psi_{\Gamma_j,d_j}}^\mathrm{het}$ and sampling~$x\sim\mu_{\psi_{\Gamma_j,d_j}}^\mathrm{het}$ can be achieved in time~$O(N^3)$. With these building blocks, Theorem~\ref{thm:main} gives an algorithm to produce a sample~$\zeta_\Psi^\mathrm{het}$ in time~$O(n\chi^2 N^3)$ except with probability at most~$\exp(-n/(\chi \|c\|_2^2))$. Taking~$n=\chi \|c\|_2^2 \log1/\delta$ gives the claim.
\end{proof}

\begin{figure}[!b]
    \begin{subfigure}[b]{0.45\textwidth}
        \includegraphics[height=4.5cm]{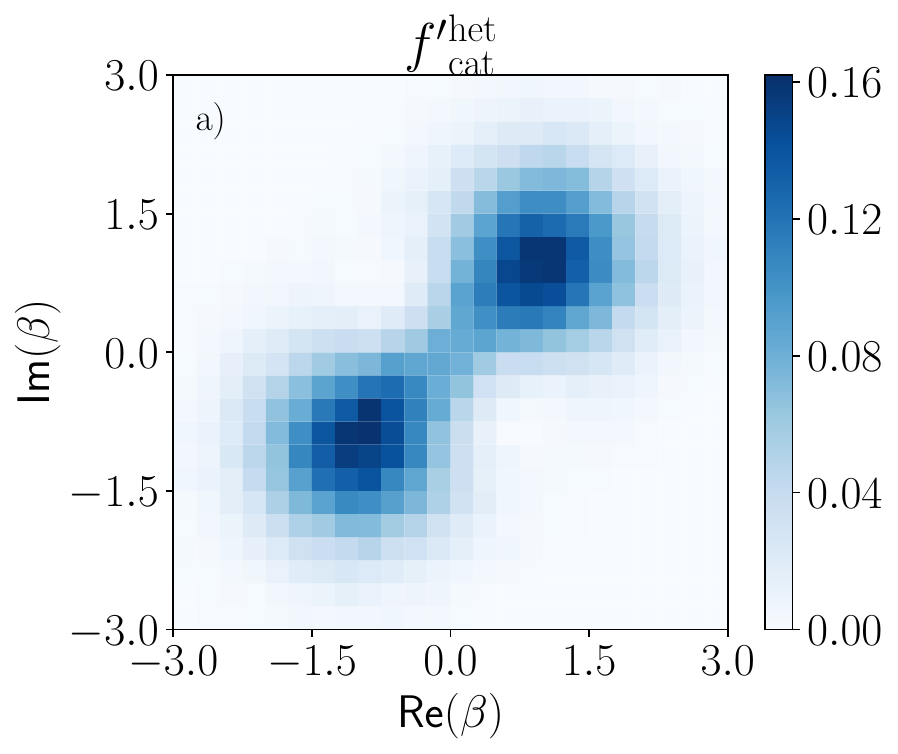}
        \caption{Normalized sample count~$f'{}^\mathrm{het}_\cat(\beta)$ (histogram) produced with~$10^5$ samples drawn from~$\mu_\cat^\mathrm{het}$ using a variant of the algorithm in \cref{cor:gaussiansuperpositionsampling} that repeats until success (i.e., acceptance). The histogram has a bin width of $0.25$. The empirically observed average number of trials until success is $1.97$. This matches the expected number of iterations~$\chi  \|c\|_2^2 \approx 1.96$ as predicted by \cref{eq:avg_number_samples} together with~\cref{lem:upperbound-dist}.
        \justifying}
    \end{subfigure}
    \begin{subfigure}[b]{0.45\textwidth}
        \includegraphics[height=4.5cm]{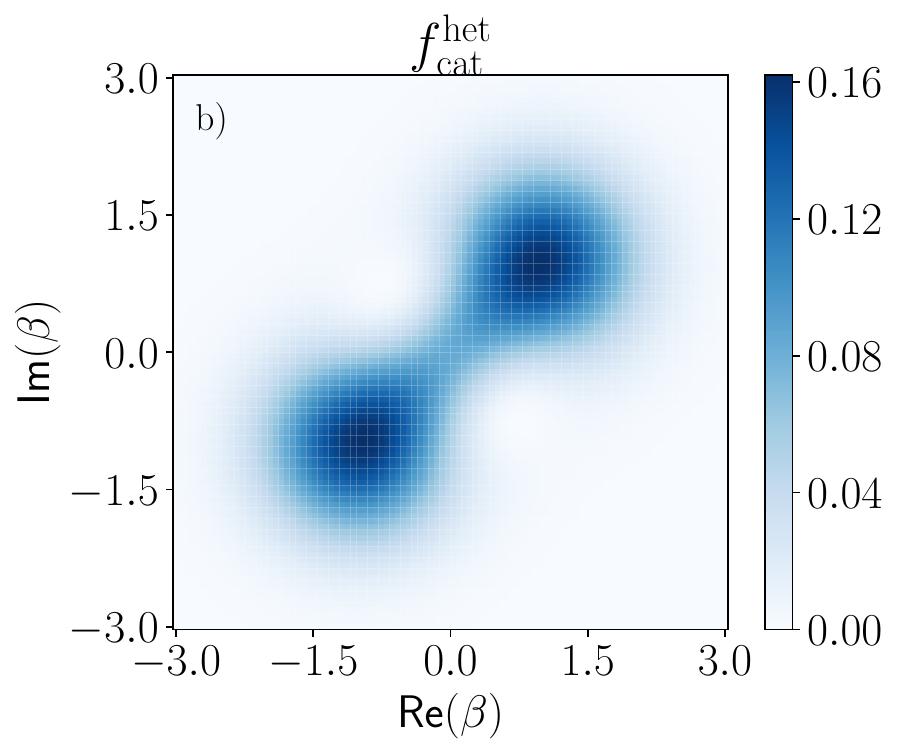}
        \caption{Probability density function~$f^\mathrm{het}_\cat(\beta)$. \justifying}
    \end{subfigure}
    \begin{subfigure}[b]{0.45\textwidth}
        \hspace{1.5mm}\includegraphics[height=4.5cm]{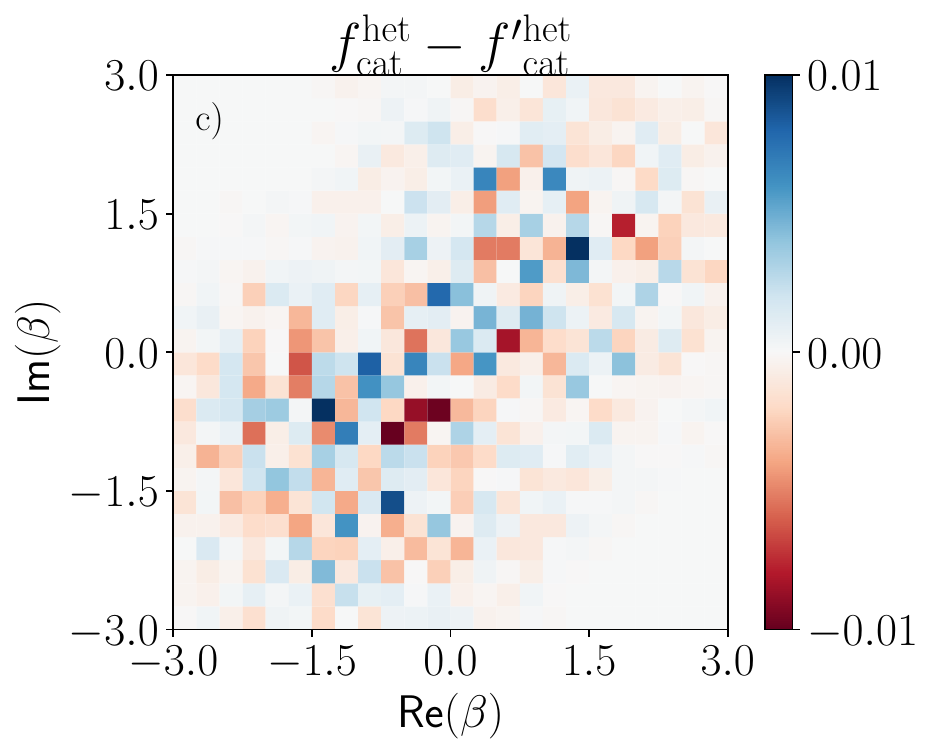}
        \caption{Difference between~$f^\mathrm{het}_\cat(\beta)$ and~$f'{}^\mathrm{het}_\cat(\beta)$. \justifying}
    \end{subfigure}
    \caption{A heterodyne measurement applied to the cat-state~$\ket{\cat}\propto \ket{\alpha} + \ket{-\alpha}$, where~$\ket{\alpha}$ denotes the coherent state with~$\alpha = 1 + i$, has outcome distribution~$\mu^\mathrm{het}_\cat$ with probability density function~$f^\mathrm{het}_\cat(\beta) = |\braket{\beta,\cat}|^2 / \pi$.  \justifying}
    \label{fig:cat}
\end{figure}

An application of  Corollary~\ref{cor:gaussiansuperpositionsampling} in given Fig.~\ref{fig:cat}. Theorem~\ref{thm:main} also provides a way of simulating homodyne measurements in a similar manner, see  Fig.~\ref{fig:gkp} for an example. As we discuss in Appendix~\ref{sec:rank-extent},
 the quantity~$\chi^3\|c\|_2^2$ in Theorem~\ref{thm:main} can be replaced by a function of the Gaussian extent~$\xi(\Psi)$, a measure of the degree of non-Gaussianity of~$\ket{\Psi}$ used in the context of classical simulation of non-Gaussian states and dynamics. We note that in the latter context, Bourassa et al.~\cite{PRXQuantum.2.040315} have proposed an alternative rejection-sampling method
for the outcome distribution of a linear optics measurement applied to 
a state whose Wigner function is a linear combination of Gaussian functions. This includes the case of a superposition of Gaussian states as considered in Corollary~\ref{cor:gaussiansuperpositionsampling}. In our opinion, the runtime guarantee of the latter (cf.~also Appendix~\ref{sec:rank-extent}) is  more easily applicable than the bound given in~\cite{PRXQuantum.2.040315}. More importantly, the reasoning underlying~\cite{PRXQuantum.2.040315} heavily depends on the particular (in this case Gaussian) functions considered, whereas Theorem~\ref{thm:main} applies more generally in black-box (oracular) settings, both in terms of function evaluation and samples from individual distributions.

\begin{figure}[!t]
    \includegraphics[width=0.8\linewidth]{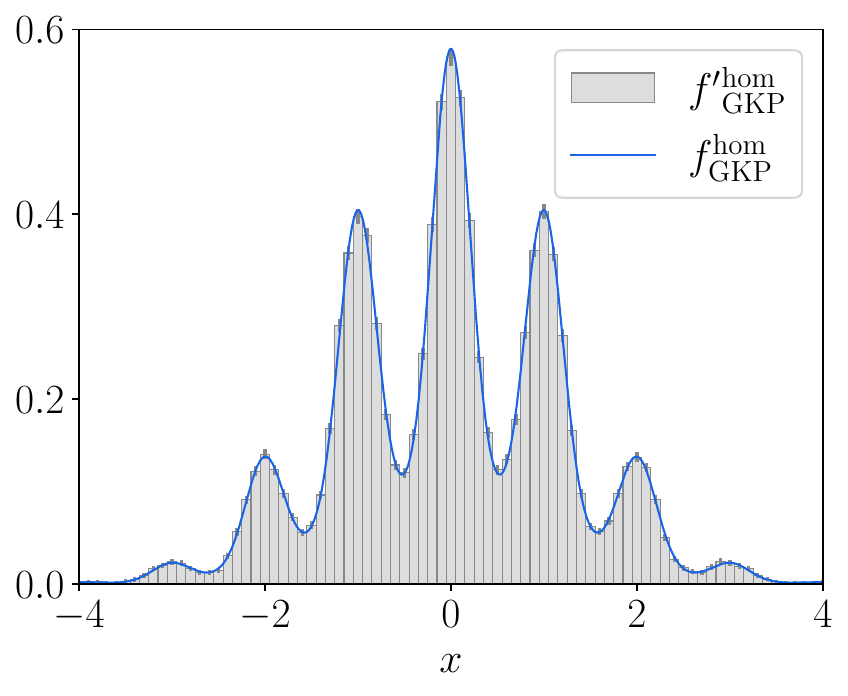}
    \caption{A homodyne measurement applied to the GKP-state~$
    \ket{\GKP} \propto \sum\nolimits_{z\in\mathbb{Z}} \exp(-\kappa^2 z^2/2) \ket{\psi(z,\Delta)}$ with~$k=0.6$ and~$\Delta=0.3$, where~$\ket{\psi_{z,\Delta}}$ is a displaced squeezed vacuum state with wavefunction~$\psi_{z,\Delta}(x) = \exp(-(x-z)^2/(2\Delta^2)) / (\pi \Delta^2)^{1/4}$, has outcome distribution~$\mu^\mathrm{hom}_\GKP$ with probability density function~$f^\mathrm{hom}_\GKP(x)$ (shown as a solid blue line). The histogram, which has a bin width of $0.1$, is the normalized sample count~$f'{}^\mathrm{hom}_\GKP(x)$ produced with~$10^5$ samples drawn from~$\mu^\mathrm{hom}_\GKP$ using a variant of the algorithm in \cref{thm:main} that repeats until success (i.e., acceptance). The empirically observed average number of trials until success is $13.46$. This matches the expected number of iterations~$\chi \|c\|_2^2 \approx 13.47$ as predicted by \cref{eq:avg_number_samples} together with~\cref{lem:upperbound-dist}.
    \justifying
    \label{fig:gkp}}
\end{figure}

The proof of Theorem~\ref{thm:main} is an  application of rejection-sampling. 
Recall that to sample from a distribution~$\nu$
 given the ability to sample from a distribution~$\mu$, this assumes that we have a uniform upper bound~$K \geq\left\|\frac{d \nu}{d \mu}\right\|_{\infty}$ on the Radon-Nikodym derivative. Given~$X_1, \ldots, X_n \sim \mu$, rejection-sampling accepts~$X_r$ with probability~$\frac{1}{K} \cdot \frac{d \nu}{d \mu}\left(X_r\right)$ 
 for each~$r\in [n]$,  and returns an arbitrary accepted~$X_r$ as a sample from~$\nu$. It can be shown that any accepted sample has distribution~$\nu$, and a single sample is accepted with probability at least~$1/K$ independently of the other samples. In particular, $n=\Theta(K)$ samples are sufficient to produce an accepted sample with high probability, see Appendix~\ref{sec:rejectionsamplingapproach}  for  details.

 In our case, the target distribution is~$\nu:=\mu^{M}_{\Psi}$. We set~$X_r:=X^M_{J_r}$. By definition, each sample~$X_r$ has distribution~$\overline{\mu}:=\sum_{j=0}^{\chi-1} p(j)\mu^M_{\psi_j}$.
 By assumption, for any~$X_r=x\in \Omega$ and~$J_r=r$, the  Radon-Nikodym derivative
 \begin{align}
 \frac{d\overline{\mu}}{d\lambda}(x)&=\sum_{j=0}^{\chi-1} p(j) \frac{d\mu^M_{\psi_j}}{d\lambda}(x)=\sum_{j=0}^{\chi-1} p(j) f^M_{\Psi_j}(x)\label{eq:radonnlmubar}
 \end{align}
 can be evaluated by~$\chi$ calls to the oracle~$p(\cdot)$ and one call to each oracle~$f^M_{\Psi_j}(\cdot)$, $j\in \{0, \ldots, \chi-1\}$. 
 Similarly,~$\frac{d\nu}{d\lambda}(x)=f^M_{\Psi}(\cdot)$ can be evaluated with one call to~$f^{M}_{\Psi}(\cdot)$. It follows that (given~$K$), the acceptance probability
 \begin{align}
 \frac{1}{K}\cdot \frac{d\nu}{d\overline{\mu}}(x)=\frac{1}{K} \frac{d\nu}{d\lambda}(x)\cdot
\left(\frac{d\overline{\mu}}{d\lambda}(x)\right)^{-1}
  \end{align}
  can be determined from these oracle calls. The claim therefore follows from 
  the characterization of rejection-sampling (see Lemma~\ref{lem:blockpolyanskiy} in Appendix~\ref{sec:rejectionsamplingapproach}) and the following lemma giving a uniform upper bound on the Radon-Nikodym derivative.

\begin{lemma}
\label{lem:upperbound-dist}
Let~$\overline{\mu}:=\sum_{j=0}^{\chi-1} p(j)\mu^M_{\psi_j}$. 
Then
\begin{align}
\label{eq:ineqaux}
\frac{d\mu^M_{\Psi}}{d\overline{\mu}}(x)&\leq \chi \|c\|_2^2\qquad\textrm{ for any }x\in \Omega\ .
\end{align}
\end{lemma}
\noindent In the special case where the states~$\{\ket{\psi_j}\}_{j=0}^{\chi-1}$ are pairwise orthogonal, ~\cref{eq:ineqaux} is a direct consequence of the following so-called pinching (operator) inequality. 
\begin{claim}[\cite{MasahitoHayashi_2002}, see Appendix~\ref{sec:appendixpinching} for a proof]\label{lem:upperboundm}
Let~$|\Psi\rangle=\sum_{j=0}^{\chi-1} c_j\left|\psi_j\right\rangle$ be a state with~$\langle \psi_j, \psi_k \rangle = 0$ for all~$j,k\in\{0,\ldots, \chi-1\}$.
Define~$P_j=\proj{\psi_j}$. Then 
\begin{align}
    |\Psi\rangle\langle\Psi| &\leq \chi \sum_{i=0}^{\chi-1} P_j|\Psi\rangle\langle\Psi| P_j=\chi \sum_{j=0}^{\chi-1}p(j)\cdot \ket{\psi_j}\bra{\psi_j}\ .
\end{align}
\end{claim}
\noindent Indeed, pairwise orthogonality of~$\{\ket{\psi_j}\}_{j=0}^{\chi-1}$ 
implies that~$\|c\|_2=1$ and~$p(j)=|\langle \psi_j,\Psi\rangle|^2$, thus~\cref{lem:upperboundm} and~$Q(x)\geq 0$ mean that  
\begin{align}
    f^M_\Psi(x)\leq \chi\|c\|_2^2 \sum_{j=0}^{\chi-1}p(j)f^M_{\psi_j}(x)\ \textrm{ for all }x\in \Omega\ .\label{eq:pdensityfncts}
\end{align}
Because of Eq.~\eqref{eq:radonnlmubar}, \cref{eq:pdensityfncts}
implies Lemma~\ref{lem:upperbound-dist}.

Next, we show \cref{eq:pdensityfncts} in the general case where the states~$\{\ket{\psi_j}\}_{j=0}^{\chi-1}$ are not pairwise orthogonal. The main idea is to orthogonalize the states~$\ket{\psi_j}$ and to apply \cref{lem:upperboundm}.
Let~$\{\ket{j}\}_{j=0}^{\chi-1}$ be an orthonormal basis of~$\mathbb{C}^\chi$ and let
$\{\ket{\widehat{k}}\}_{k=0}^{\chi-1}$ be the associated Fourier basis, where 
$\ket{\widehat{k}}=\frac{1}{\sqrt{\chi}}\sum_{j=0}^{\chi-1}e^{2\pi i jk/\chi}\ket{j}$.  Given~$|\Psi\rangle=\sum_{j=0}^{\chi-1} c_j\left|\psi_j\right\rangle\in \cH$, consider the state
\begin{align}
|\Psi'\rangle=\frac{1}{\|c\|_2}\sum_{j=0}^{\chi-1} c_j\ket{\psi_j'}
\end{align} 
where~$\ket{\psi'_j}= \left|\psi_j\right\rangle\otimes\ket{j}\in\cH\otimes\mathbb{C}^\chi$. Observe that the states~$\{\ket{\psi_j'}\}_j$ are pairwise orthogonal and thus 
\begin{align}
    \label{eq:qjprime}
    p'(j) &= |\langle \psi'_j,\Psi'\rangle|^2=\frac{|c_j|^2}{\|c\|_2^2} = p(j)\ .
\end{align}
For~$X\in \Sigma$ and~$Y\subseteq \{0,\ldots,\chi-1\}$, the expression
\begin{align}
    M'(X\times Y)&=\sum_{k\in Y }M(X)\otimes \proj{\widehat{k}}
\end{align}
defines a POVM~$M':\Sigma\times 2^{\{0,\ldots,\chi-1\}}\rightarrow\cB(\cH\otimes \mathbb{C}^\chi)$ which satisfies the Radon-Nikodym property with~$Q'(x,k)=Q(x)\otimes \proj{\widehat{k}}$ and the product measure~$\lambda'(X\times Y)=\lambda(X)\cdot |Y|$. Writing~$f^{M'}_{\Psi}(x,k)=\langle \Psi,Q'(x,k)\Psi\rangle$ and using~$|\langle\widehat{k},j\rangle|^2=1/\chi$ we have  for~$j,k\in \{0,\ldots,\chi-1\}$ and~$x\in\Omega$ 
\begin{align}
\frac{1}{\chi}f^M_{\psi_j}(x)&=  |\langle \widehat{k}, j\rangle|^2\cdot   \langle \psi_j,Q(x)\psi_j\rangle=f^{M'}_{\psi'_j}(x,k)\ .     \label{eq:pjprime}  
\end{align}
Similarly using~$\langle \widehat{0},j\rangle=\frac{1}{\sqrt{\chi}}$ we obtain 
\begin{align}
    \frac{f_{\Psi}^M(x)}{\chi \|c\|_2^2} 
    &= f_{\Psi'}^{M'}(x,0) \\
    &\leq \chi \sum_{j=0}^{\chi-1} p'(j) f_{\psi'_j}^{M'}(x,0) 
    && \textrm{by~\cref{lem:upperboundm}} \\
    &= \sum_{j=0}^{\chi-1} p(j) f_{\psi_j}^M(x) 
    && \textrm{by   
    Eqs.~\eqref{eq:qjprime},~\eqref{eq:pjprime}} \ 
    \end{align}
    concluding the proof of~\cref{eq:pdensityfncts} and hence Lemma~\ref{lem:upperbound-dist}.

Having established Theorem~\ref{thm:main}, let us briefly comment on the broader context where it may find further applications.  Depending on the measurement and the states considered, it may be computationally easy to evaluate the associated probability density functions, yet unclear how to (efficiently) sample from the corresponding distributions. In such a scenario, a quantum device allowing to prepare and measure individual states~$\{\ket{\psi_j}\}_{j=0}^{\chi-1}$ could be used to produce i.i.d.~samples as required by Theorem~\ref{thm:main}. Applying our algorithm, the quantum device could thus be used to sample from~$\mu^M_\Psi$ for a superposition~$\ket{\Psi}=\sum_{j=0}^{\chi-1}c_j\ket{\psi_j}$ even though it may not be possible to easily prepare~$\ket{\Psi}$. This hybrid use of Theorem~\ref{thm:main} -- where a quantum device is combined with efficient classical processing -- is made possible by the black-box formulation of our result: It makes no references to how the i.i.d.~samples are produced.

We note that  the ability to evaluate probabilities associated with measuring (often referred to as strong simulation) can often be leveraged to produce samples from the corresponding distribution (i.e., weak simulation). In particular, this is the case for finite-outcome measurements. A corresponding reduction from strong to weak simulation can be given along the lines sketched in the description of the generic routine~$\mathsf{samp}_p$: it suffices to compute conditional probabilities of individual bits, which can then be sampled iteratively. 
In the context of multiqubit systems and computational basis measurements,
such reductions from strong to weak simulation, as well as approximate versions thereof were previously considered e.g., in Ref.~\cite{Pashayan2020fromestimationof}. They are also implicit in various classical simulation algorithms for quantum dynamics (see e.g.,~\cite{Bravyi:2012aa,PhysRevA.65.032325}). Other approaches (applicable to output states of circuits) avoid computing marginal probabilities, instead requiring amplitude of mid-circuit states~\cite{PhysRevLett.128.220503}. More recently, a new reduction from strong to weak simulation was given in~\cite{bravyi2025classicalquantumalgorithmscharacters}.

Theorem~\ref{thm:main} can also be interpreted
as a new reduction from strong to weak simulation. For example, 
Corollary~\ref{cor:gaussiansuperpositionsampling} extends the simulation algorithm of~\cite{PhysRevA.110.042402} (see also \cite{hahn2024classicalsimulationquantumresource}) from strong to weak simulation. The key difference to the existing reductions mentioned above is that the latter only apply to finite-outcome POVMs. In contrast, Theorem~\ref{thm:main} also covers cases with a continuous number of outcomes. In particular, it circumvents more naive approaches based on discretizing (binning), which generally lead to probabilities that are less tractable and/or cannot be computed efficiently.

\textit{Acknowledgments.}
B.D. and R.K. gratefully acknowledge support by the European Research Council under Grant No. 101001976 (project EQUIPTNT).

\bibliography{references}
\bibliographystyle{unsrt}

\appendix

\section{On Radon-Nikodym theorems for POVMs\label{sec:radonnikodymPOVM}}

Let~$(\Omega,\Sigma)$ be a measure space.
A positive operator-valued measure (POVM) is
a function~$M:\Sigma\rightarrow\cB(\cH)$ 
assigning a non-negative operator~$M(X)$ to every element~$X\in \Sigma$
subject to the conditions~$M(\emptyset)=0$, $M(\Omega)=I_{\cH}$ 
and satisfying~$\sigma$-additivity, i.e., $M\left(\bigcup_{j}X_j\right)=\sum_{j}M(X_j)$
for any countable family~$\{X_j\}_j\subset \Sigma $ of pairwise disjoint sets, with the sum converging 
in the weak operator topology.

The Radon-Nikodym property concerns the representability of an additive function on a~$\sigma$-algebra.
Following~\cite{robinson2014operator}, we call a POVM~$M$ {\em decomposable}
if there is a~$\sigma$-finite measure~$\lambda$ on~$(\Omega,\Sigma)$ 
and a
weakly 
$\lambda$-measurable  function~$Q:X\rightarrow\cB(\cH)$ with~$Q(x)\geq 0$ for every~$x\in X$ such that
\begin{align}
M(X)&=\int_{X}Q(x)d\lambda(x)\quad\textrm{ for every }\quad X\in \Sigma\ .
\end{align}
In other words,~$\Omega$ can be covered by countably many measurable sets of finite measure,  and~$\langle Q(\cdot)\psi,\varphi\rangle$ is~$\lambda$-measurable for all~$\psi,\varphi\in\cH$.  In the Ph.D. thesis~\cite{robinson2014operator} by Robinson, the following necessary and sufficient condition 
for decomposability was established, see~\cite[Theorem 3.3.2]{robinson2014operator} and the remarks following it.
The result of~\cite{robinson2014operator} generalizes an earlier result by
Berezanskii and
Kondratev~\cite{berezansky1995spectral} applicable only to the case where~$M(\Omega)$ is trace-class,
and the result of Chiribella et al.~\cite{Chiribellaetal10} which shows that decomposability is given if the Hilbert space~$\cH$ is finite-dimensional. 
\begin{theorem}\cite[Theorem 3.3.2, paraphrased]{robinson2014operator}
Let~$(\Omega,\Sigma)$ be a measure space and~$M:\Sigma\rightarrow\cB(\cH)$ a POVM on a separable Hilbert space~$\cH$. 
Let~$\lambda$ be a~$\sigma$-finite measure on~$(\Omega,\Sigma)$.
Then the following are equivalent:
\begin{enumerate}[(i)]
\item  $\|M(X)\|\leq \mu(X)$ for every~$X\in \Sigma$.
\item
There is a weakly~$\lambda$-measurable map~$Q:\Omega'\rightarrow\cB(\cH)$, 
defined on a set~$\Omega'\subset \Omega$ of full~$\lambda$-measure, with~$Q(x)\geq 0$ $\mu$-almost-everywhere in~$x$, such that
\begin{align}
M(X)&=\int_X Q(x)d\lambda(x)\quad\textrm{ for all }\quad X\in \Omega\ .
\end{align}
\end{enumerate}
\end{theorem}

Important special cases are when~$\Omega=\mathbb{N}$: Here~$\lambda$ can be defined by
$\lambda(\{j\}) = M(\{j\})$ for~$j\in\mathbb{N}$.
On the other hand, if~$\Omega=\bigcup_{j\in\mathbb{N}}X_j$ is a countable union of subsets~$\{X_j\}_{j\in\mathbb{N}}\subset \Sigma$ with~$M(X_j)$ trace-class for each~$j\in\mathbb{N}$, then~$\lambda$ can be chosen as~$\lambda(X) = \tr(M(X))$. 

We refer to the more recent paper~\cite{Lili17} for additional references and  an in-depth review of the Radon-Nikodym property, including more general statements applicable to general Banach spaces (including non-separable Hilbert spaces, see e.g.,~\cite[Theorem 2.7]{Lili17}).

\section{Gaussian measurements \label{sec:gauss_measurement}}

Consider a Gaussian state~$\rho = \rho(\Gamma, d)$ on~$N$ bosonic modes with covariance matrix~$\Gamma$ and displacement vector~$d$.
A general Gaussian measurement corresponds to projecting the system into a pure Gaussian state~$\ket{\psi_G} = S(z) \ket{\alpha}$, with~$z \in [0,+\infty)^{N}$ and~$\alpha \in \bb{C}^N$, where~$S=\diag(z_1,1/z_1, \ldots, z_N,1/z_N)$ is the symplectic matrix associated with the squeezing operator~$S(z)$. The state~$\ket{\psi_G}$ has covariance matrix~$SS^T$ and displacement vector~$m = S d(\alpha)$. 
The probability density function defined by the measurement is \cite{PhysRevLett.89.137903,PhysRevA.66.032316} (see also e.g., \cite{Sanchez2007QuantumIW})
\begin{align}
    \label{eq:pm_gaussian}
    f( m ) = \frac{\exp\left(  -(m - d)^T (\Gamma + SS^T)^{-1} (m - d) \right)}{\pi^N \sqrt{\det((\Gamma + SS^T)/2)}}\ .
\end{align}

Note that a heterodyne Gaussian measurement corresponds to having~$z = (1,\ldots,1)$ and~$S=I$, while a homodyne measurement corresponds to taking the limit~$z \rightarrow (0,\ldots,0)$. Notice that
\begin{align}
    \lim_{z \rightarrow 0} \left(\Gamma + SS^T\right)^{-1} = (X\Gamma X)^{MP}
\end{align}
where~$X=\diag(1,0,1,0,\ldots)$ and~$MP$ denotes the Moore-Penrose inverse. 
The probability distribution for the outcome~$m = (x_1,0,\ldots,x_N,0)$ (where~$x_j$ is the outcome of a homodyne measurement on mode~$j$) is 
\begin{align}
    f( m ) = \frac{\exp\left(  -(m - d)^T (X\Gamma X)^{MP} (m - d) \right)}{\pi^N \sqrt{\det((X\Gamma X)/2)}} \ .
\end{align}

\section{Relation to the Gaussian extent of a state \label{sec:rank-extent}}

Consider a state~$\ket{\Psi}\in L^2(\mathbb{R}^N)$ of the form
\begin{align}
\ket{\Psi}&=\sum_{j=0}^{r-1}c_j \ket{\omega_j}\ ,
\end{align}
where~$r\in\mathbb{N}$ is arbitrary,~$\{\ket{\omega_j}\}_{j=0}^{r-1}$ are Gaussian states, and~$\{c_j\}_{j=0}^{r-1}$ are complex coefficients such that~$\ket{\Psi}$ is normalized. The so-called ``sparsification lemma'' of  Ref.~\cite{PhysRevLett.116.250501} roughly states the following (when translated to our setting): 
 if~$\|c\|_1^2$ is small, then the state~$\ket{\Psi}$ has a sparse approximation, i.e., there is a vector~$\ket{\Omega}\in L^2(\mathbb{R}^N)$  (not necessarily normalized) which
 \begin{enumerate}[(i)]
 \item is a linear combination of a small number of Gaussian states, and 
 \item is also close to~$\ket{\Psi}$. 
 \end{enumerate}
 We use the formulation given 
in \cite[Lemma 6]{Bravyi2019simulationofquantum}, which provides additional structural information on the form of~$\ket{\Omega}$. More precisely, \cite[Lemma 6]{Bravyi2019simulationofquantum} together with Markov's inequality imply that 
for~$\varepsilon>0$ and any integer 
\begin{align}
\chi\geq \frac{2\|c\|_1^2}{\varepsilon^2}\label{eq:chilowerboundcdelta}
\end{align}
there is a (unnormalized) vector~$\ket{\Omega}$ of the form
\begin{align}
\ket{\Omega} &=\frac{\|c\|_1}{\chi}\sum_{j=0}^{\chi-1}\ket{\psi_j}\ ,
\end{align} 
where each state~$\ket{\psi_j}$ for~$j\in \{0,\ldots,\chi-1\}$ is a Gaussian state (in fact, an element of~$\{\ket{\omega_k}\}_{k=0}^{r-1}$ up to a phase factor), and
\begin{align}
\|\Psi-\Omega\|\leq \varepsilon\ .\label{eq:normdifference}
\end{align}
(In \cite{Bravyi2019simulationofquantum}, it is shown that a randomly chosen element~$\ket{\Omega}$ from a certain ensemble of states satisfies these properties with probability at least~$1/2$. In fact, an additional stronger concentration bound is shown (giving a higher probability), however, this requires additional information about~$\ket{\Psi}$ in the form of a quantity called the stabilizer fidelity.)

We note that~\eqref{eq:normdifference} implies that
\begin{align}
\left|1-\|\Omega\|\right|\leq \varepsilon\ .\label{eq:normbounddist}
\end{align}
\begin{proof}
Because~$\|\Psi\|=1$ we have 
\begin{align}
\|\Psi-\Omega\|^2&= 1-2\mathsf{Re}\langle \Psi,\Omega\rangle+\|\Omega\|^2 \leq \varepsilon^2\ 
\end{align}
and
\begin{align}
\mathsf{Re}\langle \Psi,\Omega\rangle &\leq
 \left|\mathsf{Re}\langle \Psi,\Omega\rangle\right|\\
 &\leq |\langle\Psi,\Omega\rangle|\\
 &\leq \|\Omega\|
\end{align}
by the Cauchy-Schwarz inequality. This implies the claim.
\end{proof}

Consider the normalized state
\begin{align}
\ket{\Psi'}&=\frac{1}{\|\Omega\|}\ket{\Omega}\ .
\end{align}
By the triangle inequality as well as~\eqref{eq:normdifference} and~\eqref{eq:normbounddist}, we have
\begin{align}
\|\Psi-\Psi'\|&\leq \|\Psi-\Omega\|+
\|\Omega-\Psi'\|\\
&\leq \varepsilon+\left|\|\Omega\|-1\right| \\
\label{eq:psipsiprimedif} &\leq 2\varepsilon\ .
\end{align}

By definition, we have
\begin{align}
\ket{\Psi'}&=\sum_{j=0}^{\chi-1} c_j' \ket{\psi_j}\ \textrm{ where } c_j'=\frac{\|c\|_1}{\chi \|\Omega\|}\textrm{ for every }j\ .
\end{align}
With~\eqref{eq:chilowerboundcdelta} and~\eqref{eq:normbounddist} it follows that 
\begin{align}
    \|c'\|_2^2 = \frac{\|c\|_1^2}{\chi \|\Omega\|^2}
    \leq \frac{\varepsilon^2}{2\|\Omega\|^2}
    \leq \frac{\varepsilon^2}{2(1-\varepsilon)^2} \ .
\end{align}
In particular, we conclude that
\begin{align}
    \label{eq:cprimenormbound} \|c'\|_2^2\leq 2\varepsilon^2\quad\textrm{ for }\quad \varepsilon<1/2 
    \ . 
\end{align}

For a state~$\ket{\Psi}\in L^2(\mathbb{R}^N)$, the Gaussian extent~$\xi(\Psi)$  is defined as the infimum of~$\|c\|_1^2$ over all decompositions~$\ket{\Psi}=\sum_{j=0}^{r-1}c_j\ket{\omega_j}$ into Gaussian states~$\{\ket{\omega_j}\}_{j=0}^{r-1}$, where~$r\in\mathbb{N}$ is arbitrary (see e.g.,~ \cite{PhysRevA.110.042402,hahn2024classicalsimulationquantumresource}).
Considering a sequence of decompositions~$\ket{\Psi}=\sum_{j=0}^{r-1}c_j\ket{\omega_j}$ approaching this infimum, the above can be summarized as follows:
\begin{lemma}\label{lem:sparsificationmodified}
Let~$\ket{\Psi}\in L^2(\mathbb{R}^N)$. Let~$\varepsilon<1/2$. For any~$\chi>2\xi(\Psi)/\varepsilon^2$ (where the inequality is strict) there is a state
\begin{align}
\ket{\Psi'}&=\sum_{j=0}^{\chi-1} c_j' \ket{\psi_j}
\end{align}
where~$\{\ket{\psi_j}\}_{j=0}^{\chi-1}$ are Gaussian states such that 
\begin{align}
\|\Psi-\Psi'\|&\leq 2\varepsilon
\end{align}
and
\begin{align}
\|c'\|_2\leq \sqrt{2}\varepsilon\ .
\end{align}
\end{lemma}
For concreteness, we define a state~$\ket{\Psi'}$ as in Lemma~\ref{lem:sparsificationmodified} with~$\chi=3\xi(\Psi)/\varepsilon^2$ an~$\varepsilon$-sparsification of~$\Psi$. 
Without loss of generality (by adjusting the phase of~$c_j'$), we may assume that~$\ket{\psi_j}=\ket{\psi_{\Gamma_j,d_j}}$ is the unique
Gaussian state with covariance matrix~$\Gamma_j$ and displacement vector~$d_j$ satisfying~$\langle \mathrm{\mathrm{vac}}, \psi_{\Gamma_j,d_j}\rangle>0$, for every~$j\in \{0,\ldots,\chi-1\}$.  Applying Theorem~\ref{thm:main} we obtain the following in a scenario where an~$\varepsilon$-sparsification of a state~$\ket{\Psi}$ is known, i.e., the parameters~$\{(c'_j,\Gamma_j,d_j)\}_{j=0}^{\chi-1}$ are given.
\begin{corollary}
\label{cor:gaussiansuperpositionsampling_extent}
Let~$\ket{\Psi}\in L^2(\mathbb{R}^N)$. Let~$\varepsilon<1/2$. There is a  probabilistic algorithm which takes as input a tuple~$\{(c'_j,\Gamma_j,d_j)\}_{j=0}^{\chi-1}$
describing an~$\varepsilon$-sparsification~$\ket{\Psi'}$ of the state~$\ket{\Psi}$, and outputs a sample~$\zeta_\Psi^M$ from a distribution~$8\varepsilon$-close to~$\mu_\Psi^M$ in~$L_1$-norm with probability at least~$1-\delta$. 
The algorithm has runtime upper bounded by~$O(N^3 \xi(\Psi)^3 \varepsilon^{-4} \log1/\delta)$.
\end{corollary}

\begin{proof}
We have
\begin{align}
    &|\mu_\Psi^M(X) - \mu_{\Psi'}^M(X)| \\
    &= |\braket{\Psi, M(X) \Psi} - \braket{\Psi', M(X) \Psi'}| \\
    &\leq |\braket{\Psi, M(X) \Psi} - \braket{\Psi, M(X) \Psi'}| \\
    &\qquad + | \braket{\Psi, M(X) \Psi'} - \braket{\Psi', M(X) \Psi'}| \\
    &\leq |\langle \Psi,M(X)(\Psi-\Psi')\rangle|+|\langle \Psi-\Psi',M(X)\Psi'\rangle \\
    &\leq 2\|\Psi-\Psi'\| \\
    &\leq 4 \varepsilon
\end{align}
for any~$X\in\Sigma$,
where we applied the triangle and the Cauchy-Schwarz inequalities, and where we used~$M(X)\leq I$ and \cref{eq:psipsiprimedif}. 
As a consequence
\begin{align}
\|\mu_{\Psi'}^M - \mu_{\Psi}^M\|_1 = 2 
        \sup_{X\in\Sigma}| \mu_{\Psi'}^M(X) - \mu_{\Psi}^M(X) |
        \leq 8 \varepsilon \ .
\end{align}
Therefore the algorithm in \cref{cor:gaussiansuperpositionsampling} with input~$\{(c'_j,\Gamma_j,d_j)\}_{j=0}^{\chi-1}$ produces a sample from a distribution~$8\varepsilon$-close to~$\mu_{\Psi}^M$ in~$L_1$-norm with success probability at least~$1-\delta$ in time~$O(N^3\chi^3\|c'\|_2^2 \log1/\delta)$. From the definition of an~$\varepsilon$-sparsification we have~$\chi = 3\xi(\Psi)/\varepsilon^2$. This together with \cref{eq:cprimenormbound} gives the upper bound for the runtime in the claim.  
\end{proof}

\section{Rejection-sampling\label{sec:rejectionsamplingapproach}}
We consider a measure space~$(\Omega,\Sigma)$ where~$\Sigma$ denotes a~$\sigma$-algebra on the set~$\Omega$. Rejection sampling is a well-established technique, which can be succinctly summarized as follows:
\begin{lemma}\cite[Lemma 26, paraphrased]{BlockPolyanskiy23}\label{lem:blockpolyanskiy}
Let~$\mu, \nu$ be two measures on~$(\Omega, \Sigma)$ and suppose that~$X \sim \mu$. Suppose that~$\mu, \nu$ are such that~$\left\|\frac{d \nu}{d \mu}\right\|_{\infty} \leq M$ for some~$M<\infty$. 
Let~$\xi$ be a binary-valued random variable such that conditioned on~$X$, the probability that~$\xi=1$ is~$\frac{1}{M} \frac{d \nu}{d \mu}(X)$. Then~$\Pr(\xi=1)=\frac{1}{M}$ and
$$
\Pr(X \in A \mid \xi=1)=\nu(A)\qquad\textrm{ for any }A\in \Sigma\ .
$$
In particular, if~$X_1, \ldots, X_n$ are sampled independently and identically with~$X_r\sim \mu$, and~$\zeta_1, \ldots, \zeta_n$ are constructed as above,
then at least one of~$\{\zeta_j\}_{j=1}^n$ is equal to~$1$ with probability at least~$1-e^{-\frac{n}{M}}$.
\end{lemma}
We refer to~\cite{BlockPolyanskiy23} for a proof of this statement.

It is worth noting that the average number of samples required to obtain one successful event (which corresponds to obtaining $\xi=1$) is given by 
\begin{align}
    &\bb{E}[\text{number of samples until success}] \\
    &\qquad= \Pr(\xi=1) \sum_{j=1}^\infty j \left(1-\Pr(\xi=1)\right)^{j-1} \\
    \label{eq:avg_number_samples} &\qquad= \frac{1}{\Pr(\xi=1)} = M \ .
\end{align}

\section{Proof of the pinching inequality\label{sec:appendixpinching}}
For completeness, we include a proof of Claim~\ref{lem:upperboundm}, i.e., the pinching inequality established in~\cite{MasahitoHayashi_2002}.
\begin{proof}
Define~$V_j=\sum_{j=0}^{\chi-1} e^{2 \pi i j k / \chi} P_k$. 
For any operator~$\rho\geq 0$ we have 
\begin{align}
    \sum_{j=0}^{\chi-1} V_j \rho V_j^\dagger&=\sum_{j=0}^{\chi-1} \sum_{k=0}^{\chi-1} \sum_{\ell=0}^{\chi-1} e^{2 \pi i j(k-\ell) / \chi} P_k \rho P_\ell \\
    & =\sum_{k,\ell=0}^{\chi-1}\Bigg(\underbrace{\sum_{j=0}^{\chi-1} e^{2 \pi i j(k-\ell) / \chi}}_{\chi \delta_{k,\ell}}\Bigg) P_k \rho P_\ell \\
    & =\chi \sum_{k=0}^{\chi-1} P_k \rho P_k \ .
\end{align}
For~$\rho \geq 0$  we have 
\begin{align}
    V_{j }\rho V_j^\dagger \geq 0\qquad\textrm{ for all }j \in \{0,\ldots, \chi-1\}
\end{align}
hence
\begin{align}
    V_0 \rho V_0^\dagger+\underbrace{\sum_{j=1}^{\chi-1} V_j \rho V_j^\dagger}_{\geq 0}=\chi \sum_{k=1}^{\chi-1} P_k \rho_k P_k\ .
\end{align}
It follows that 
\begin{align}
    V_0 \rho V_0^\dagger \leq  \chi \sum_{k=0}^{\chi-1} P_k \rho P_k\ .
\end{align}

Take~$\rho=\proj{\Psi}$. It remains to show that~$|\Psi\rangle\langle\Psi| = V_0|\Psi\rangle\langle\Psi| V_0^\dagger$.
We have
\begin{align}
    V_0|\Psi\rangle\langle\Psi| V_0^\dagger&=\sum_{\ell, m=0}^{\chi-1} c_\ell \overline{c_m}  V_0\left|\psi_\ell\right\rangle\left\langle\psi_m\right| V_0^\dagger\ .
\end{align}
By definition, we have 
\begin{align}
    V_0&=\sum_{k=0}^{\chi-1} P_k = P_\ell+Q_\ell
\end{align}
for every~$\ell\in \{0,\ldots,\chi-1\}$, where~$Q_\ell:=\sum_{k=0, k\neq \ell }^{\chi-1} P_k$.
It follows that
\begin{align}
    V_0|\Psi\rangle\langle\Psi| V_0^\dagger&=\sum_{\ell, m=0}^{\chi-1} c_\ell \overline{c_m}  (P_\ell+Q_\ell)\left|\psi_\ell\right\rangle\left\langle\psi_m\right| (P_m+Q_m)\\
    &=\left(\ket{\Psi}+\ket{\Theta}\right)\left(\bra{\Psi}+\bra{\Theta}\right)
\end{align}
where we defined
\begin{align}
    \ket{\Theta}&=\sum_{\ell=0}^{\chi-1}Q_\ell c_\ell \ket{\psi_\ell}\ .
\end{align}
The claim follows since~$\ket{\Theta}=0$ because the states~$\ket{\psi_\ell}$ are orthogonal. 

\end{proof}

\end{document}